\documentclass[11pt]{article}

\usepackage{amsmath,amsthm,amssymb,mathtools,authblk,braket,cite}

\usepackage[a4paper]{geometry}
\theoremstyle{plain}
\newtheorem{theorem}{Theorem}[section]

\newtheorem{corollary}[theorem]{Corollary}
\newtheorem{lemma}[theorem]{Lemma}
\newtheorem{proposition}[theorem]{Proposition}

\theoremstyle{definition}

\theoremstyle{remark}

\newcommand{\ii}{\mathrm{i}}
\newcommand{\e}{\mathrm{e}}
\newcommand{\tr}{\operatorname{tr}}

\newcommand{\R}{\mathbb{R}}

\newcommand{\Dom}[0]{{\mathrm{Dom}\,}}

\newcommand{\hilb}{\mathcal{H}}

\usepackage[unicode=true,
bookmarks=true,bookmarksopen=false,
breaklinks=false,pdfborder={0 0 0},colorlinks=true]
{hyperref}
\usepackage{xcolor}
\definecolor{cblue}{rgb}{0.16, 0.32, 0.75}
\definecolor{cred}{rgb}{0.7, 0.11, 0.11}
\hypersetup{%
	,linkcolor=cred
	,citecolor=cblue
	,urlcolor=black
}

\usepackage{orcidlink}

\title{\textbf{Positive Hamiltonians cannot give \\exponential decay of positive observables}}

\author[$\,1,3$]{\normalsize Paolo Facchi\hspace{3pt}\orcidlink{0000-0001-9152-6515}}
\author[$\,2,3,\star$]{\normalsize Davide Lonigro\hspace{3pt}\orcidlink{0000-0002-0792-8122}}

\affil[$1$]{\small Dipartimento di Fisica, Universit\`a di Bari, I-70126 Bari, Italy}
\affil[$2$]{\small Dipartimento di Matematica, Universit\`a di Bari, I-70125 Bari, Italy}
\affil[$3$]{\small Istituto Nazionale di Fisica Nucleare, Sezione di Bari, I-70126 Bari, Italy}
\affil[$\star$]{\small \href{mailto:davide.lonigro@ba.infn.it}{\texttt{davide.lonigro@ba.infn.it}}}
\date{\normalsize\today}

\begin{document}
	
\vspace{-1.5cm}
\maketitle\vspace{-0.5cm}
	
\begin{abstract}
The survival probability of a quantum system with a finite ground energy is known to decay subexponentially at large times. Here we show that, under the same assumption, the average value of any quantum observable, whenever well-defined, cannot converge exponentially to an extremal value of the spectrum of the observable. Large-time deviations from the exponential decay are therefore a general feature of quantum systems. As a simple application of these results, we show that, when considering an open quantum system whose dynamics is generated by a Hamiltonian with a finite ground energy, a large-time exponential decay of populations is forbidden, whereas coherences may still decay exponentially.
\end{abstract}

\section{Introduction}
The description of decay phenomena has been a central topic in quantum mechanics since its inception. A fundamental role is played by exponential decay: unstable systems are often experimentally shown to follow an exponentially decaying law at large times. Quantum mechanical models exhibiting such behaviors were first provided in~\cite{Gurney,Gamov,Weisskopf}.

There is, however, a huge limitation in the quantum mechanical description of exponential decay: whenever the Hamiltonian generating the dynamics of the system is bounded from below, i.e.\ admits a finite ground energy, the survival probability of any pure state of the system cannot decay exponentially at large times. Its evolution must follow a slower decay law, for instance a power law. This is the content of a highly influential no-go argument first obtained by Khalfin~\cite{Khalfin57,Khalfin58} as a direct consequence of the celebrated Paley-Wiener theorem in Fourier analysis~\cite{PaleyWiener}: exponential decay at large times requires, as a necessary conditions, that the Hamiltonian has a doubly unbounded spectrum, which is regarded as unphysical in quantum theory. Khalfin's argument can be summarized as follows: given a vector state $\psi$ in a Hilbert space $\hilb$, with $\|\psi\|=1$, and a Hamiltonian $H$ bounded from below, the survival amplitude of $\psi$ under the evolution generated by~$H$,
\begin{equation}\label{eq:surv}
	t\in\mathbb{R}\mapsto \braket{\psi|\e^{-\ii tH}\psi}\in\mathbb{C},
\end{equation}
must satisfy the following inequality:
\begin{equation}\label{eq:paley}
	\int_{-\infty}^\infty\frac{\log\left|\braket{\psi|\e^{-\ii tH}\psi}\right|}{1+t^2}\,\mathrm{d} t>-\infty.
\end{equation}
Consequently, the survival probability of any state cannot behave exponentially at large times: if it decays to zero, its decay should be slower than any exponential function $\e^{-\gamma t}$.
 Mathematically, this follows from the fact that, if the spectrum of $H$ is bounded from below, the function in Eq.~(\ref{eq:surv}) can be extended to a bounded analytic function in the lower complex half-plane, continuous up to the real axis, which must necessarily satisfy Eq.~(\ref{eq:paley}).
Further details can be found, for example, in~\cite{Peres,Newton,Fonda,Gaemers,Pascazio96,Exner,Giacosa14} and references therein. 

Since most Hamiltonians in quantum physics are bounded from below, exponential decay at large times is thus argued to be \textit{unphysical}, even though deviations from the exponential law are usually hard to detect experimentally: the first observations of sub-exponential behavior at large times are relatively recent~\cite{Crespi}.
However, the survival probability is not the only physically interesting decay law associated with a quantum system: we are often interested in the time behavior of the average value of quantum observables. Khalfin's argument does not extend directly to a generic quantum average. Indeed, the evolution law of the average of some observables \textit{can} be exponential. This was first pointed out in~\cite{Facchi17}: the reduced dynamics of a bipartite system can be Markovian and some local observables can decay with a pure exponential law. It is thus natural to ask whether exponential decay of quantum averages can still be ruled out for some class of observables. 

In this paper we will show that Khalfin's argument can be greatly generalized: positive Hamiltonians \textit{cannot} give purely exponential decay of the average value of a \textit{positive} observable. The result holds for bounded observables as well as unbounded ones, provided that the average value is well-defined at every time. 

The time evolution of the average value on $\psi$ of any positive observable $A$, whenever well-defined, is given by
\begin{equation}
	\braket{A}_\psi(t)= \braket{\e^{-\ii tH}\psi| A\, \e^{-\ii tH}\psi} = \|A^{1/2}\e^{-\ii tH}\psi\|^2.
\end{equation}
We will show that the function
\begin{equation}
	t\in\mathbb{R}\mapsto A^{1/2}\e^{-\ii tH}\psi\in\hilb
\end{equation}
must satisfy an inequality analogous to Eq.~(\ref{eq:paley}):
\begin{equation}\label{eq:paley2}
\int_{-\infty}^\infty\frac{\log\|A^{1/2}\e^{-\ii tH}\psi\|}{1+t^2}\,\mathrm{d} t>-\infty,
\end{equation}
ruling out exponential decay. 
Analogous estimates can be obtained, more generally, in the case in which the system is initially assumed to be prepared in a (possibly) mixed state, represented by a density operator $\varrho\in\mathcal{B}(\hilb)$---that is, a positive trace-class operator with unit trace.

More generally, this means that, if the average value of a quantum observable $A$ converges to an \textit{extremal} value of its spectrum, then the convergence cannot be exponential: an exponential law to a value $a$ can only be obtained as the result of a \textit{compensation} between the contributions of the positive and negative parts of the spectrum of $A-a$. There is nothing inherently ``special" about the survival probability of pure states, which corresponds to choosing $A=\ket{\psi}\!\bra{\psi}$: deviations from the exponential law are a much more general feature of quantum mechanical systems. 

This paper is organized as follows. In Section~\ref{sec:1} we show that the evolution law of a state $\psi$ induced by a Hamiltonian $H$ bounded from below admits a vector-valued analytic extension to the lower half-plane; in Section~\ref{sec:2} we exploit this formalism to generalize Khalfin's argument to $\braket{A}_\psi(t)$; finally, in Section~\ref{sec:3} we discuss some applications to open quantum system theory, and in Section~\ref{sec:4} we summarize our results. The appendix is devoted to the proof of a logarithmic inequality for vector-valued analytic functions, which is crucial for proving our results. 

\section{Analytic continuation of the evolution group}\label{sec:1}
Let $H$ be a self-adjoint operator on a separable Hilbert space $\hilb$, with scalar product $\braket{\cdot |\cdot}$ and norm $\|\cdot\|$. By Stone's theorem, $H$ is uniquely associated with a time-homogeneous \textit{unitary propagator}, i.e.~a strongly continuous unitary group $\{U(t)\}_{t\in\mathbb{R}}$ on $\hilb$, defined via
\begin{equation}
\e^{-\ii tH}=\int_{\mathbb{R}} \e^{-\ii t\lambda}\,\mathrm{d}P_H(\lambda),
\end{equation}
with $P_H$ being the projection-valued measure associated with $H$. By construction, for all $\psi\in\hilb$, the vector-valued function
\begin{equation}\label{eq:evol}
t\in\mathbb{R}\mapsto\e^{-\ii tH}\psi\in\hilb
\end{equation}
is continuous. Physically, for a quantum system, given any state $\psi\in\hilb$, the map in~(\ref{eq:evol}) represents the \textit{evolution} of the state of a quantum system which, at $t=0$, is in the state $\psi$. We remark that, for all $\alpha\geq0$,
\begin{equation}\label{eq:moment}
\e^{-\ii tH}\Dom H^\alpha=\Dom H^\alpha,
\end{equation}
$\Dom H^\alpha$ being the domain of $H^\alpha$, and that the evolution map~(\ref{eq:evol}), for a generic $\psi\in\hilb$, is continuous but \textit{not} necessarily differentiable. Differentiability holds if and only if $\psi\in\Dom H$, with
\begin{equation}
\ii\frac{\mathrm{d}}{\mathrm{d}t}\e^{-\ii tH}\psi=H\e^{-\ii tH}\psi,
\end{equation}
and more generally, whenever $\psi\in\Dom H^k$ for $k\in\mathbb{N}$,
\begin{equation}
\left(\ii\frac{\mathrm{d}}{\mathrm{d}t}\right)^k\e^{-\ii tH}\psi=H^k\e^{-\ii tH}\psi = \e^{-\ii tH} H^k\psi.
\end{equation}
The regularity of the map in~(\ref{eq:evol}) is thus directly linked with the number of energy moments that are finite in the state $\psi$, namely $\|H^k \psi\|<\infty$.

From now on we will suppose $H$ to be \textit{bounded from below}. Without loss of generality, we will take $H\geq0$. The spectrum of $H$ is thus contained in $[0,\infty)$ and, for all $t\in\mathbb{R}$, we can write
\begin{equation}
\e^{-\ii tH}=\int_0^\infty\e^{-\ii\lambda t}\,\mathrm{d}P_H(\lambda).
\end{equation}
As a crucial consequence, the following property holds:
\begin{proposition}\label{prop:evolana}
	Let $H\geq 0$. For all $\tau\in\mathbb{C}^-$, the open lower half-plane of complex numbers, the operator
	\begin{equation}\label{eq:eitau}
	\e^{-\ii\tau H}=\int_0^\infty\e^{-\ii\lambda\tau}\,\mathrm{d}P_H(\lambda)
	\end{equation}
	is well-defined and bounded. Besides, for all $\psi\in\hilb$, we have
	\begin{equation}
	\e^{-\ii\tau H}\psi\in\bigcap_{\alpha\geq0}\Dom H^{\alpha} ,\label{eq:regular}
	\end{equation}
	and the map
	\begin{equation}\label{eq:complexder}
	\tau\in\mathbb{C}^-\cup\mathbb{R}\mapsto\e^{-\ii\tau H}\psi\in\hilb
	\end{equation}
	is a bounded vector-valued function, analytic in $\mathbb{C}^-$ and continuous up to the real axis, with complex derivative
	\begin{equation}\label{eq:complexder1}
	\left(\ii\frac{\mathrm{d}}{\mathrm{d}\tau}\right)^k\e^{-\ii\tau H} \psi =H^k\e^{-\ii\tau H}\psi ,
	\end{equation}
	for all $k\in\mathbb{N}$ and all $\tau\in\mathbb{C}^-$.
\end{proposition}
\begin{proof}
	For all $\psi\in\hilb$, setting $\tau=t-\ii\eta$ with $\eta\geq0$, we have
	\begin{equation}\label{eq:bounded}
	\|\e^{-\ii\tau H}\psi\|^2=\int_0^\infty|\e^{-\ii\lambda\tau}|^2\,\mathrm{d}\mu_\psi(\lambda)=\int_0^\infty\e^{-2\eta\lambda}\,\mathrm{d}\mu_\psi(\lambda)\leq \|\psi\|^2<\infty,
	\end{equation}
	since the function $\lambda\in[0,\infty)\mapsto\e^{-2\lambda\eta}$ is bounded by 1, and $\mu_\psi(\cdot)=\|P_H(\cdot)\psi\|^2$ is a finite measure with $\mu_\psi(\R)=\|\psi\|^2$. Consequently, $\e^{-\ii\tau H}$ is a bounded operator. 
	
	Moreover, we also have, for all $\alpha\geq0$, and $\eta>0$,
	\begin{equation}
	\| H^\alpha\e^{-\ii\tau H}\psi\|^2=\int_0^\infty\lambda^{2\alpha}|\e^{-\ii\lambda\tau}|^2\,\mathrm{d}\mu_\psi(\lambda)=\int_0^\infty\lambda^{2\alpha}\e^{-2\eta\lambda}\,\mathrm{d}\mu_\psi(\lambda)<\infty,
	\end{equation}
	implying that $H^{\alpha}\e^{-\ii\tau H}$ is itself a bounded operator, or equivalently that, for all $\psi\in\hilb$, $\e^{-\ii\tau H}\psi\in\Dom H^{\alpha}$.
	The latter property also means that the function~(\ref{eq:complexder}) can be differentiated (in the strong sense) in $\mathbb{C}^{-}$, thus implying analyticity and, for all $k\in\mathbb{N}$, Eq.~(\ref{eq:complexder1}). Boundedness of the function $\tau\mapsto\e^{-\ii\tau H}\psi$ follows from Eq.~(\ref{eq:bounded}).
\end{proof}
Consequently, whenever the spectrum is bounded from below, the evolution map can be continuously extended to a function in lower complex half-plane which is analytic for \textit{every choice} of $\psi$. Hence, the choice of $\psi$ only affects the regularity property of its behavior \textit{at the boundary}, i.e.\ on the real line. In other words, the ``nice" properties of the analytic continuation of $\tau\mapsto\e^{-\ii\tau H}\psi$ follow because, when extended to the lower half-plane, the evolution group acquires an exponentially decaying term and hence ``regularizes" all states on which it acts.

As an immediate corollary of Proposition~\ref{prop:evolana}, given $\alpha\geq0$, for all $\psi\in\Dom H^{\alpha}$ the function
\begin{equation}
\tau\in\mathbb{C}^-\cup\mathbb{R}\mapsto H^{\alpha}\e^{-\ii\tau H}\psi\in\hilb
\end{equation}
is a well-defined bounded function, analytic in $\mathbb{C}^-$ and continuous up to $\mathbb{R}$. If we take $\psi\notin\Dom H^{\alpha}$, we still obtain a well-defined analytic map in $\mathbb{C}^-$, but the latter cannot be extended to the real axis nor it is bounded.

\section{Evolution of quantum averages}\label{sec:2}
Recall that quantum observables are associated with self-adjoint operators $A$ on $\hilb$. In particular, let $A$ be a positive self-adjoint operator on $\hilb$, with domain $\Dom A$. The average value $\braket{A}_\psi$ of $A$ at the state $\psi$ is well-defined if and only if $\psi\in\Dom A^{1/2}$, which is the \textit{form domain} of $A$, and reads
\begin{equation}
\braket{A}_\psi=\|A^{1/2}\psi\|^2.
\end{equation}
However, the average value of $A$ at any time $t\in\mathbb{R}$, i.e.
\begin{equation}
\braket{A}_\psi(t)=\|A^{1/2}\e^{-\ii tH}\psi\|^2,\qquad t\in\mathbb{R},
\end{equation}
will generally be \textit{ill-defined} unless some minimal requirements on the relation between $A$ and the Hamiltonian $H$ are made. Indeed, even if $\psi\in\Dom A^{1/2}$, in general $\e^{-\ii tH}\psi$ need not be in $\Dom A^{1/2}$ for all $t\in\mathbb{R}$. A minimal requirement ensuring well-definiteness of $\braket{A}_\psi(t)$ for all times is the following:
\begin{lemma}\label{lemma2}
	Let $A$ and $H$ be positive self-adjoint operators on $\hilb$ such that
	\begin{equation}\label{eq:formdomains}
	\Dom A^{1/2}\supset\Dom H^{\alpha},
	\end{equation}
	for some $\alpha\geq0$. Then, for all $\psi\in\Dom H^{\alpha}$, the map
	\begin{equation}\label{eq:avmap}
	t\in\mathbb{R}\mapsto A^{1/2}\e^{-\ii tH}\psi\in\hilb
	\end{equation}
	is well-defined.
\end{lemma}
\begin{proof}
	Immediate consequence of Eq.~(\ref{eq:moment}).
\end{proof}
Property~(\ref{eq:formdomains}) means that the form domain of $A$ contains the form domain of the $\alpha$th power of $H$, i.e., physically, that all states with well-defined $\alpha$th moment of the energy also have a well-defined average value of $A$. Notice that, if~(\ref{eq:formdomains}) holds for some $\alpha$, then it also holds for all $\beta\geq\alpha$: the smaller the minimal value of $\alpha$ for which~(\ref{eq:formdomains}) holds, the larger is the linear subspace for which the map~(\ref{eq:avmap}) is well-defined. In particular, $\alpha=0$ corresponds to the case of a bounded observable, for which no domain issues arise and~(\ref{eq:avmap}) exists for all $\psi\in\hilb$.

Recalling Proposition~\ref{prop:evolana}, we can now provide an analytic continuation of the function $t\in\mathbb{R}\mapsto A^{1/2}\e^{-\ii tH}\psi\in\hilb$:
\begin{proposition}\label{prop:ossana}
		Let $A$ and $H$ be positive self-adjoint operators on $\hilb$ satisfying $\Dom A^{1/2}\supset\Dom H^{\alpha}$ for some $\alpha\geq0$. Then, for all $\psi\in\Dom H^{\alpha}$, the map
	\begin{equation}
	\tau\in\mathbb{C}^-\cup\mathbb{R}\mapsto A^{1/2}\e^{-\ii\tau H}\psi\in\hilb
	\label{eq:mapa1}
	\end{equation}
	is a bounded vector-valued function, analytic in $\mathbb{C}^-$ and continuous up to the real axis.
\end{proposition}
\begin{proof}
		By Proposition
		~\ref{prop:evolana} (see Eq.~(\ref{eq:regular})) and the assumption $\Dom A^{1/2}\supset\Dom H^{\alpha}$, for all $\tau\in\mathbb{C}^-$ we have $\e^{-\ii\tau H}\psi\in\Dom A^{1/2}$, hence the map in~(\ref{eq:mapa1}) is well-defined.
	
	To prove the remaining properties, notice that, since both $A^{1/2}$ and $H^{\alpha}$ are self-adjoint operators, the property $\Dom A^{1/2}\supset\Dom H^{\alpha}$ automatically implies that $A^{1/2}$ is relatively bounded with respect to $H^{\alpha}$ (see e.g.~\cite[Lemma~6.2]{Teschl}), i.e.~there exist $a,b\geq0$ such that, for all $\psi\in\Dom H^{\alpha}$,
	\begin{equation}
	\|A^{1/2}\psi\|\leq a\|H^{\alpha}\psi\|+b\|\psi\|.
	\end{equation}
	This allows us to prove boundedness of the map $\tau\mapsto A^{1/2}\e^{-\ii\tau H}\psi$, since
	\begin{eqnarray}
	\|A^{1/2}\e^{-\ii\tau H}\psi\|&\leq&a\|H^{\alpha}\e^{-\ii\tau H}\psi\|+b\|\e^{-\ii\tau H}\psi\|\nonumber\\
	&\leq&a\|H^{\alpha}\psi\|+b \|\psi\|.
	\end{eqnarray}
	Continuity on $\mathbb{R}$ and analyticity in $\mathbb{C}^-$ follow similarly from the continuity and analyticity of $\tau\mapsto\e^{-\ii\tau H}\psi$, $H^{\alpha}\e^{-\ii\tau H}\psi$, and the bounds above.
\end{proof}
We are now ready to state our result:
\begin{theorem} \label{thm:genPW}
Let $A$ and $H$ be positive operators on $\hilb$ satisfying $\Dom A^{1/2}\supset\Dom H^{\alpha}$ for some $\alpha\geq0$. Then for all $\psi\in\Dom H^{\alpha}$, we have
\begin{equation}
	\int_{-\infty}^{\infty}\frac{\log\|A^{1/2}\e^{-\ii tH}\psi\|}{1+t^2}\,\mathrm{d}t>-\infty.
	\label{eq:khalfinh}
\end{equation}
\end{theorem}
\begin{proof}
	By Proposition~\ref{prop:ossana}, the function $\tau\in\mathbb{C}^-\cup\mathbb{R}\mapsto F(\tau)=A^{1/2}\e^{-\ii\tau H}\psi\in\hilb$ is bounded, analytic in $\mathbb{C}^-$ and continuous up to $\mathbb{R}$. As shown in Lemma~\ref{lm:genPW} in the appendix, such function must satisfy a logarithmic inequality
	\begin{equation}\label{eq:logarithmic2}
	\int_{-\infty}^{\infty}\frac{\log\|F(t)\|}{1+t^2}\,\mathrm{d}t>-\infty,
	\end{equation}
	which is Eq.~(\ref{eq:khalfinh}) in our case.
\end{proof}
Consequently, the map $t\mapsto \braket{A}_\psi(t)=\|A^{1/2}\e^{-\ii tH}\psi\|^2$ \textit{cannot} be exponential at large times, i.e., if
\begin{equation}
\lim_{t\to\infty}\braket{A}_\psi(t)=0,
\end{equation}
then the convergence must be subexponential. We have thus generalized Khalfin's argument to the average value of any positive observable, whenever it is well-defined. 

Obviously, there is nothing special about positive observables: if $A$ is bounded from below with $a\in\mathbb{R}$ being the infimum of its spectrum, then $A\geq a$, and the argument shows that the average value of $A$ cannot converge exponentially to $a$. Moreover, if $A$ is bounded from above, then $-A$ is bounded from below. We can therefore conclude, as claimed in the Introduction, that the average value of a quantum observable cannot converge exponentially to a finite \textit{extremal} value of its spectrum.

 Exponential convergence to a non-extremal value of the spectrum, instead, is not ruled out in general. Indeed, let $A$ be a quantum observable and $\psi$ satisfying $\braket{A}_\psi(t)\to0$. In the most general case, we can always decompose $A$ in its domain as
\begin{equation}
A=A_+-A_-,\qquad A_\pm\geq0,
\end{equation}
and thus
\begin{equation}
\braket{A}_\psi(t)=\braket{A_+}_\psi(t)-\braket{A_-}_\psi(t).
\end{equation}
Suppose that both $\braket{A_\pm}_\psi(t)$ converge to $0$ as well: by our argument, the decay laws of $A_\pm$ cannot be exponential at large times, but, in principle, slower decaying terms of $A_+$ and $A_-$ may cancel out, thus yielding a purely exponential decay of $\braket{A}_\psi(t)$. Explicit examples of this phenomenon were given in~\cite{Facchi17}. This cannot happen if either $A_+=0$ or $A_-=0$, i.e.~when the limiting value $0$ is an extremal point of the spectrum of $A$.

In general, our argument allows us to conclude that, \textit{if} the average value of a positive $A$ decays to zero, then it cannot decay exponentially. As a concrete application of our argument to a family of cases in which we know that the average decays, we will provide the following corollary. Recall that $A$ is said to be \textit{relatively compact} with respect to $H$ if the operator $A\,R_H(z)$, with $R_H(z)$ being the resolvent of $H$ at $z\in\mathbb{C}$, is compact for some $z$. Besides, relative compactness of $A$ implies $\Dom A\supset\Dom H$.
\begin{corollary}
	Let $A$ and $H$ be positive self-adjoint operators such that $A^{1/2}$ is \textit{relatively compact} with respect to $H$. Then, for every $\psi\in\Dom H$ in the absolutely continuous subspace of $H$, the map
	\begin{equation}
		t\in\mathbb{R}\mapsto\|A^{1/2}\e^{-\ii tH}\psi\|^2\in\mathbb{R}
	\end{equation}
	decays \textit{subexponentially} as $t\to\infty$.
\end{corollary}
\begin{proof}
Under our hypotheses it is known (see e.g.~\cite[Theorem 5.6]{Teschl}) that $t\in\mathbb{R}\mapsto\|A^{1/2}\e^{-\ii tH}\psi\|^2$ converges to zero. Theorem~\ref{thm:genPW} then holds with $\alpha=1$.
\end{proof}

We finally remark that similar results can be straightforwardly inferred for higher moments of $A$: if we suppose $\Dom A^{\delta/2}\supset\Dom H^{\alpha}$ for some $\delta,\alpha\geq0$, then the average value of $A^\delta$, i.e.~the $\delta$th moment of $A$ will satisfy an analogous constraint and thus, in particular, will not undergo any exponential decay.

\section{Extension to mixed states, and an application to open quantum systems}\label{sec:3}

As a starting point, notice that the discussion above may be easily replicated, \textit{mutatis mutandis}, to the case in which the quantum system starts its evolution from a possibly \textit{mixed} state, represented by a positive trace-class operator $\varrho\in\mathcal{B}(\hilb)$ of trace one. For our purposes, it will be sufficient to restrict our attention to the following case:
\begin{equation}\label{eq:sqrt}
	\varrho=\sum_{j=1}^rp_j\ket{\psi_j}\!\bra{\psi_j},
\end{equation}
with $r$ being a finite integer, $p_1,\dots,p_r$ being a family of nonnegative real numbers summing to one, and $\psi_1,\dots,\psi_r\in\hilb$ of unit norm. We shall again consider a positive operator such that $\Dom A^{1/2}\supset\Dom H^\alpha$ for some $\alpha\geq0$. Then, as an immediate consequence of Lemma~\ref{lemma2}, if $\psi_j\in \Dom H^\alpha$, the map
\begin{equation}\label{eq:hsmap}
	t\in\mathbb{R}\mapsto A^{1/2}\e^{-\ii tH}\sqrt{\varrho}\in\mathrm{HS}
\end{equation}
is well-defined. Here, $\mathrm{HS}$ is the space of Hilbert-Schmidt operators on $\hilb$, which is itself a Hilbert space with the scalar product $\Braket{T|S}_{\rm HS}:=\tr T^\dag S$. In this case, the average value of $A$ is defined at all times:
\begin{equation}
	\Braket{A}_\varrho(t)=\tr\left[A\e^{-\ii tH}\varrho\,\e^{\ii tH}\right]=\left\|A^{1/2}\e^{-\ii tH}\sqrt{\varrho}\right\|^2_{\rm HS}.
\end{equation}
\begin{corollary}\label{coroll2}
	Let $A$ and $H$ be positive self-adjoint operator on $\hilb$ satisfying $\Dom A^{1/2}\supset\Dom H^\alpha$ for some $\alpha\geq0$, and let $\varrho\in\mathcal{B}(\hilb)$ be a density operator as in~(\ref{eq:sqrt}), such that $\psi_1,\dots,\psi_r\in\Dom H^\alpha$. Then
	\begin{equation}
		\int_{-\infty}^{\infty}\frac{\log\left\|A^{1/2}\e^{-\ii tH}\sqrt{\varrho}\,\right\|_{\rm HS}}{1+t^2}\,\mathrm{d}t>-\infty.
		\label{eq:khalfinhbis}
	\end{equation}
\end{corollary}
\begin{proof}
	As a direct consequence of Proposition~\ref{prop:ossana}, the map~(\ref{eq:hsmap}) admits a bounded $\mathrm{HS}$-valued extension to $\mathbb{C}^-\cup\mathbb{R}$ which is analytic in $\mathbb{C}^-$ and continuous up to the real axis. Whence the claim follows analogously as in the pure state case by simply replacing the norm of $\hilb$ with the Hilbert--Schmidt norm.
\end{proof}
Consequently, the discussion in the previous section holds without substantial differences if we consider a system whose dynamics starts from a mixed state $\varrho$, modulo some minor technical assumptions ensuring the operator $A^{1/2}\e^{-\ii tH}\sqrt{\varrho}$ to be well-defined and Hilbert-Schmidt at all time.

As a simple application of this result, let us consider a bipartite quantum system, represented by a Hilbert space $\hilb=\hilb_{\rm S}\otimes\hilb_{\rm E}$. Physically, $\hilb_{\rm S}$ is the space associated with the experimentally accessible system, with finite dimension $\dim\hilb_{\rm S}=d$, while $\hilb_{\rm E}$ represents the external, inaccessible environment. Suppose that, at the initial time, the system and the environment are in a separable (uncorrelated state)
\begin{equation}
	\varrho=\rho\otimes\ket{\Omega_0}\!\bra{\Omega_0},
\end{equation}
with $\rho\in\mathcal{B}(\hilb_{\rm S})$ being a density matrix, and $\Omega_0\in\hilb_{\rm E}$. The reduced evolution induced on the system by a global Hamiltonian is thus represented by the map
\begin{equation}
	t\in\mathbb{R}^+\mapsto \Lambda_t(\rho)=\tr_{\rm E}\left[\e^{-\ii tH}\rho\otimes\ket{\Omega_0}\!\bra{\Omega_0}\e^{\ii tH}\right],
\end{equation}
with $\tr_E$ denoting the partial trace with respect to the environment's degrees of freedom. Given two orthonormal bases $\{\xi_j\}_{j=0,\dots,d-1}\subset\hilb_{\rm S}$, $\{\Omega_n\}_{n\in\mathbb{N}}\subset\hilb_{\rm E}$, one clearly has, for all $j$,
\begin{eqnarray}
	\Braket{\xi_j|\Lambda_t(\rho)\xi_j}&=&\braket{\xi_j|\tr_{\rm E}\left[\e^{-\ii tH}\rho\otimes\ket{\Omega_0}\!\bra{\Omega_0}\e^{\ii tH}\right]\xi_j}\nonumber\\
	&=&\sum_{n\in\mathbb{N}}\braket{\xi_j\otimes\Omega_n|\left(\e^{-\ii tH}\rho\otimes\ket{\Omega_0}\!\bra{\Omega_0}\e^{\ii tH}\right)\xi_j\otimes\Omega_n}\nonumber\\
	&=&\sum_{n\in\mathbb{N}}\left\|P_{jn}\e^{-\ii tH}\sqrt{\rho}\otimes\ket{\Omega_0}\!\bra{\Omega_0}\right\|^2_{\rm HS},
\end{eqnarray}
with $P_{jn} = \ket{\xi_j\otimes\Omega_n}\!\bra{\xi_j\otimes\Omega_n}$ being the rank-one projector associated with the vector $\xi_j\otimes\Omega_n$. As a direct consequence of Corollary~\ref{coroll2}, each term in the sum above, if converging to zero, vanishes more slowly than exponentially at large times, whence so does $\Braket{\xi_j|\Lambda_t(\rho)\xi_j}$.

In other words: given an open quantum system whose dynamics is generated by a global Hamiltonian $H$ with a spectrum bounded from below, then the \textit{diagonal elements} of the evolved density operator $\rho(t)=\Lambda_t(\rho)$ in any fixed basis cannot decay exponentially, a result that was first pointed out in~\cite{Beau}. Of course, this argument does \textit{not} apply to the off-diagonal elements of the density matrix, which may decay exponentially regardless the positivity of the global Hamiltonian, with the dephasing model studied in~\cite{Facchi17} (cf. also~\cite{Lonigro22}) being a clear counterexample: a positive Hamiltonian can generate a purely Markovian dephasing semigroup, involving a pure exponential decay of the off-diagonal elements of the density matrix.

Physically, the \textit{populations} of an open quantum system cannot deplete exponentially at arbitrarily large times: deviations from the exponential decay at sufficiently large times are predicted. Khalfin's argument is thus recovered as the unitary version of a more general statement. Instead, no such constraint is imposed on the \textit{coherences}, which may decay exponentially.

\section{Conclusions}\label{sec:4}
We have shown that, under quite general conditions, an exponential decay at large times of the average value of a positive quantum observable is prohibited whenever the Hamiltonian generating the dynamics has a finite ground energy, which is a physical requirement for quantum systems. An exponential decay of a quantum observable requires its spectrum to admit a nontrivial negative part: in that case, slower decaying terms generated by the positive and negative parts of the spectrum may cancel out in such a way to produce an overall exponential decay.

Our result generalizes well-known properties of the survival probability of a pure state, showing that deviations from the exponential decay are indeed a general features of quantum systems, while also clarifying precisely the conditions under which such deviations are expected to emerge. Mathematically, these results follow as a consequence of the possibility, in the presence of a Hamiltonian bounded from below, to define an analytic continuation of the evolution group from the real line to the lower complex half-plane: the evolution group is thus seen as the boundary limit of an analytic function, and this allows us to infer a constraint on its decay properties. Future works may be dedicated to infer other dynamical properties of quantum systems via a similar technique.
\bigskip

\textit{Acknowledgments}.  We acknowledge the support by PNRR MUR project CN00000013 `Italian National Centre on HPC, Big Data and Quantum Computing',
 by the Italian National Group of Mathematical Physics (GNFM-INdAM), and by Istituto Nazionale di Fisica Nucleare (INFN) through the project `QUANTUM'.
 
\appendix

\section{A logarithmic inequality}\label{app}

\renewcommand{\thesubsection}{\Alph{subsection}}
In this appendix we will prove Eq.~(\ref{eq:logarithmic2}) for an $\hilb$-valued bounded analytic function in the lower complex half-plane, continuous up to the real axis. 

Let us start from the scalar case: consider a nonzero function $f:\mathbb{C}^-\cup\mathbb{R}\rightarrow\mathbb{C}$, analytic in $\mathbb{C}^-$, continuous on the real axis, and bounded. We will denote the complex variable as $\tau\in\mathbb{C}^-\cup\mathbb{R}$. It is known that such functions must satisfy the logarithmic inequality
\begin{equation}\label{eq:logarithmic1}
	\int_{-\infty}^{\infty}\frac{\log|f(t)|}{1+t^2}\,\mathrm{d}t>-\infty.
\end{equation}
Let us briefly revise how to obtain Eq.~(\ref{eq:logarithmic1}). The standard technique is the following one:
\begin{itemize}
	\item we transform $f$ into an analytic function $\hat{f}$ in the unit disk $D_1=\{\zeta\in\mathbb{C}:|\zeta|<1\}$, continuous up to the boundary, via a conformal mapping (i.e.~an angle-preserving map) between $\mathbb{C}^-$ and $D_1$;
	\item we use a well-known property of such functions: $\log|\hat{f}|$ is subharmonic (Proposition~\ref{prop:subharmonic}), i.e.~it satisfies the inequality~(\ref{eq:jensen});
	\item we transfer this inequality back to the lower half-plane, thus obtaining Eq.~(\ref{eq:logarithmic1}).
\end{itemize}
We will sketch the main steps of this procedure. First of all, let us consider explicitly a conformal mapping between $\mathbb{C}^-$ and $D_1$. To this purpose, we can take the map
\begin{equation}\label{eq:conformal}
\tau\in\mathbb{C}^-\mapsto\hat{\zeta}(\tau)=\frac{\ii+\tau}{\ii-\tau}\in D_1,
\end{equation}
and its inverse
\begin{equation}\label{eq:conformal2}
\zeta\in D_1\mapsto\hat{\tau}(\zeta)=\ii\frac{\zeta-1}{\zeta+1}\in\mathbb{C}^-.
\end{equation}
 If we allow $\zeta$ to range on the closure of $D_1$, $\overline{D_1}=\left\{\zeta\in\mathbb{C}:|\zeta|\leq1\right\}$, then $\tau$ will range on the closed half plane $\mathbb{C}^-\cup\mathbb{R}$ plus the point at infinity $\infty$. In particular:
\begin{itemize}
	\item the real line $\mathbb{R}$ is mapped to the unit circle minus the point $\zeta=-1$;
	\item $\tau=\infty$ is mapped to the point $\zeta=-1$;
	\item $\tau=-\ii$ is mapped to the center $\zeta=0$ of the disk.
\end{itemize}
Let us transfer our problem from the half-plane to the unit disk. Define the function
\begin{equation}\label{eq:ftoh}
\hat{f}:D_1\rightarrow\mathbb{C},\qquad \hat{f}(\zeta)=f\left(\hat{\tau}(\zeta)\right)\in\mathbb{C}.
\end{equation}
By construction, $\hat{f}$ is analytic in $D_1$. Besides, since $f(\tau)$ is continuous up to the real axis (which is mapped to the circle minus $\zeta=-1$) \textit{and} is bounded, i.e.~it is also regular at $\infty$, $\hat{f}(\zeta)$ inherits the same properties of $f(\tau)$: it is analytic in the unit disk and continuous on the boundary.

 For such functions, the following inequality must hold:
\begin{proposition}\label{prop:subharmonic}
	Let the function 	$\zeta\in\overline{D_1}\mapsto \hat{f}(\zeta)\in\mathbb{C}$ be analytic in $D_1$ and continuous up to the boundary. Then $\zeta\in\overline{D_1}\mapsto\log|\hat{f}(\zeta)|\in\mathbb{R}\cup\{-\infty\}
	$
	is subharmonic, i.e.~it satisfies
	\begin{equation}\label{eq:jensen}
	\log|\hat{f}(0)|\leq\frac{1}{2\pi}\int_0^{2\pi}\log|\hat{f}(\e^{\ii\theta})|\,\mathrm{d}\theta.
	\end{equation}
\end{proposition}
\begin{proof}
	Suppose that $\hat{f}(0)\neq0$. This property is a straightforward consequence of Jensen's formula:
\begin{equation}\label{eq:jensen2}
\log|\hat{f}(0)|=\sum_{j=1}^n\log|\zeta_j|+\frac{1}{2\pi}\int_0^{2\pi}\log|\hat{f}(\e^{\ii\theta})|\,\mathrm{d}\theta.
\end{equation}
Here $\zeta_1,\dots\zeta_n$ are the zeros of $\hat{f}(\zeta)$ in $D_1$. A complete proof of this formula can be found in many books of complex analysis (see e.g.~\cite{ahlfors1966complex}). Besides, the hypothesis $\hat{f}(0)\neq0$ can be relaxed: a further generalization of Jensen's formula may be shown to hold without that requirement.
\end{proof}

As a consequence, the value of the integral on the right-hand side is bounded from below, and therefore
\begin{equation}\label{eq:khalfin0}
\int_{0}^{2\pi}\log|\hat{f}(\e^{\ii\theta})|\,\mathrm{d}\theta>-\infty.
\end{equation}
Now we can come back to the half-plane and translate this inequality in terms of the original function $f$, which, by Eq.~(\ref{eq:ftoh}), can be reconstructed via
\begin{equation}
f(\tau)=\hat{f}\left(\hat{\zeta}(\tau)\right)\in\mathbb{C}.
\end{equation}
Via conformal mapping, the integral of $\hat{f}(\zeta)$ on the unit circle corresponds to the integral of $f(\tau)$ on the real line, up to a weight term coming from the Jacobian of the transformation: Eq.~(\ref{eq:khalfin0}) is therefore equivalent to
\begin{equation}
\int_{-\infty}^{\infty}\frac{\log|f(t)|}{1+t^2}\,\mathrm{d}t>-\infty,
\end{equation}
which is the logarithmic inequality~(\ref{eq:logarithmic1}).

Let us generalize this property to the vector-valued case. Analytic functions having values in a Hilbert space $\hilb$ are a straightforward generalization of the usual notion of $\mathbb{C}$-valued holomorphic functions. The standard theory of analytic functions generalizes immediately to the Hilbert-valued case without substantial differences. With little effort, it can be shown that a $\hilb$-valued function is analytic if and only if it is \textit{weakly} analytic (see e.g.~\cite{conway2019course}), i.e.~if, for all $\phi\in\hilb$, the function $z\mapsto\braket{\phi|f(z)}\in\mathbb{C}$ is analytic. This property allows us to easily identify $\hilb$-valued analytic functions.

To show the extension of the logarithmic inequality~(\ref{eq:logarithmic2}) to $\hilb$-valued functions $F:\mathbb{C}^-\cup\mathbb{R}\rightarrow\hilb$, analytic in $\mathbb{C}^-$, continuous on the real axis and bounded, we will follow a similar strategy. First of all, again we will use the conformal mappings~(\ref{eq:conformal})--(\ref{eq:conformal2}) between the lower half-plane and the unit disk to construct an analytic, $\hilb$-valued function on the real disk:
\begin{equation}
\hat{F}:D_1\rightarrow\mathbb{C},\qquad \hat{F}(\zeta)=F(\hat{\tau}(z))\in\hilb,
\end{equation}
which, again, is analytic inside the unit disk and continuous up to the circle. Proposition~\ref{prop:subharmonic} can be generalized to the vector-valued case:
\begin{proposition}\label{khalfinh}
	Let the vector-valued function 	$\zeta\in\overline{D_1}\mapsto \hat{F}(\zeta)\in\hilb$ be analytic in $D_1$ and continuous up to the boundary. Then the function $	\zeta\in\overline{D_1}\mapsto\log\|\hat{F}(\zeta)\|\in\mathbb{R}\cup\{-\infty\}
	$
	is subharmonic, i.e.\ it satisfies
	\begin{equation}\label{eq:jensen3}
	\log\|\hat{F}(0)\|\leq\frac{1}{2\pi}\int_0^{2\pi}\log\|\hat{F}(\e^{\ii\theta})\|\,\mathrm{d}\theta.
	\end{equation}
\end{proposition}
\begin{proof}
	Since $\hat{F}$ is analytic, for every $\phi\in\hilb$ the complex-valued function $\zeta\in\overline{D_1}\mapsto\braket{\phi|\hat{F}(\zeta)}\in\mathbb{C}$ is analytic in $D_1$ and continuous up to the unit circle, therefore, by Proposition~\ref{prop:subharmonic}, it satisfies
	\begin{equation}
	\log|\braket{\phi|\hat{F}(0)}|\leq\frac{1}{2\pi}\int_0^{2\pi}\log|\braket{\phi|\hat{F}(\e^{\ii\theta})}|\,\mathrm{d}\theta.
	\end{equation}
	On the other hand, by the Riesz-Fr\'echet representation theorem one has	\begin{equation}
	\|\hat{F}(\zeta)\|=\sup_{\|\phi\|=1}|\braket{\phi|\hat{F}(\zeta)}| \ ,
	\end{equation}
	thus, also using the monotonicity of the logarithm,
	\begin{equation}
	\log\|\hat{F}(\zeta)\|=\log\sup_{\|\phi\|=1}|\braket{\phi|\hat{F}(\zeta)}|=\sup_{\|\phi\|=1}\log|\braket{\phi|\hat{F}(\zeta)}| \ .
	\end{equation}
	This means that, for all $\epsilon>0$, there is $\phi_{\epsilon}\in\hilb$ with $\|\phi_{\epsilon}\|=1$ such that
	\begin{equation}
	\log\|\hat{F}(0)\|\leq\log|\braket{\phi_{\epsilon}|\hat{F}(0)}| +\epsilon,
	\end{equation}
	and therefore
	\begin{eqnarray}
	\log\|\hat{F}(0)\|&\leq&\epsilon+\frac{1}{2\pi}\int_0^{2\pi}\log|\braket{\phi_{\epsilon}|\hat{F}(\e^{\ii\theta})}|\,\mathrm{d}\theta\nonumber\\
	&\leq&\epsilon+\frac{1}{2\pi}\int_0^{2\pi}\log\|\hat{F}(\e^{\ii\theta})\|\,\mathrm{d}\theta\nonumber.
	\end{eqnarray}
	Since $\epsilon$ is arbitrary, Eq.~(\ref{eq:jensen3}) follows.
\end{proof}
Again, this property implies that the integral in the right-hand side of Eq.~(\ref{eq:jensen3}) is bounded from below, and therefore cannot diverge to $-\infty$. Transferring back the problem from the unit disk to the lower half-plane, this finally proves the following lemma which is the sought result, see Eq.~(\ref{eq:logarithmic2}):
\begin{lemma}
\label{lm:genPW}
Let $F:\mathbb{C}^-\cup\mathbb{R}\rightarrow\hilb$ be a nonzero vector-valued function analytic in $\mathbb{C}^-$, continuous on the real axis, and bounded. Then	
\begin{equation}
\int_{-\infty}^{\infty}\frac{\log\|F(t)\|}{1+t^2}\,\mathrm{d}t>-\infty.
\end{equation}
\end{lemma}

\end{document}